\documentclass{article}

\usepackage{amsmath,amssymb,amsthm}

\usepackage{a4wide}
\usepackage[all]{xy}
\usepackage{subfig}

\usepackage{graphicx}
\usepackage{url}
\usepackage{morefloats}

\usepackage{pstricks}
\usepackage{epsfig}
\usepackage{pst-grad} 
\usepackage{pst-plot} 

\newtheorem{theorem}{Theorem}[section]

\newtheorem{proposition}[theorem]{Proposition}
\newtheorem{lemma}[theorem]{Lemma}

\theoremstyle{remark}
\newtheorem{remark}{Remark}[section]

\theoremstyle{definition}
\newtheorem{definition}{Definition}[section]

\begin{document}

\title{A SICA compartmental model in epidemiology\\ 
with application to HIV/AIDS in Cape Verde\thanks{This is a preprint 
of a paper whose final and definite form is with 
\emph{Ecological Complexity}, ISSN 1476-945X, available at 
[http://dx.doi.org/10.1016/j.ecocom.2016.12.001]. 
Submitted 29/April/2016; Revised 21/Oct/2016; Accepted 02/Dec/2016.}}

\author{Cristiana J. Silva\thanks{Corresponding author.}\\ 
\texttt{cjoaosilva@ua.pt}
\and 
Delfim F. M. Torres\\
\texttt{delfim@ua.pt}}

\date{Center for Research and Development in Mathematics and Applications (CIDMA)\\
Department of Mathematics, University of Aveiro, 3810-193 Aveiro, Portugal}


\maketitle

\begin{abstract}
We propose a new mathematical model for the transmission 
dynamics of the human immunodeficiency virus (HIV). 
Global stability of the unique endemic equilibrium
is proved. Then, based on data provided by the 
``Progress Report on the AIDS response in Cape Verde 2015'', 
we calibrate our model to the cumulative cases of infection 
by HIV and AIDS from 1987 to 2014 and we show that our 
model predicts well such reality. Finally, a sensitivity 
analysis is done for the case study in Cape Verde. 
We conclude that the goal of the United Nations 
to end the AIDS epidemic by 2030 is a nontrivial task.
\end{abstract}





\section{Introduction}

The infection by human immunodeficiency virus (HIV) continues 
to be a major global public health issue, having claimed more 
than 34 million lives so far \cite{HIV:AIDS:FactSheet2015}. 
The most advanced stage of HIV infection is acquired 
immunodeficiency syndrome (AIDS). There is yet no cure or vaccine to AIDS.
However, antiretroviral (ART) treatment improves health, prolongs life, 
and substantially reduces the risk of HIV transmission. 
In both high-income and low-income countries, the life expectancy 
of patients infected with HIV who have access to ART is now measured
in decades, and might approach that of uninfected populations in patients 
who receive an optimum treatment (see \cite{AIDS:chronic:Lancet:2013} 
and references cited therein). However, ART treatment still presents 
substantial limitations: does not fully restore health;
treatment is associated with side effects; the medications are expensive; 
and is not curative. According to the Joint United Nations Programme 
on HIV and AIDS (UNAIDS), the main advocate for accelerated, comprehensive 
and coordinated global action on the HIV/AIDS epidemic, the number 
of people living with HIV continues to increase, in large part because 
more people are accessing ART therapy globally. In 2015, 15.8 million 
people were accessing treatment.  New HIV infections have fallen by 
35\% since 2000 and AIDS-related deaths have fallen by 42\% since 
the peak in 2004. In 2000, fewer than 1\% of people living with HIV 
in low- and middle-income countries had access to treatment. In 2014, 
the global coverage of people receiving ART therapy was 40\%. 
However, in 2014, around 2 million people were newly infected with HIV 
and 1.2 million people died of AIDS-related illnesses. 
Of the 36.9 million people living with HIV globally in 2014, 
17.1 million did not know they had the virus and were not aware 
of the need to be reached with HIV testing services. 
Around 22 million do not had access to HIV treatment, 
including 1.8 million children \cite{AIDSnumbers2015}.

Cape Verde is an island country spanning an archipelago 
of 10 volcanic islands in the central Atlantic Ocean. 
Located 570 kilometres (350 miles) off the coast of Western Africa, 
the islands cover a combined area of slightly over 4,000 
square kilometres (1,500 square miles).
Since the early 1990s, Cape Verde has been a stable 
representative democracy, and remains one of the most 
developed and democratic countries in Africa. 
Its population is of around 506,000 residents with
a sizeable diaspora community across all the world, 
outnumbering the inhabitants on the islands.
Such facts make Cape Verde an important
case study, with global implications \cite{MR3175429}.

In 2014, 409 new HIV cases were reported in Cape Verde, 
accumulating a total of 4,946 cases. Of this total,
1,766 developed AIDS and 1,066 have died.
The municipality with more cases was Praia, 
followed by Santa Catarina (Santiago island) and S\~{a}o Vicente. 
Cape Verde has developed a Strategic National Plan to fight against AIDS, 
which includes ART treatment, monitoring of patients, 
prevention actions and HIV testing. From the first diagnosis 
of AIDS in 1986, Cape Verde got significant progress in the fight, 
prevention and treatment of HIV/AIDS \cite{report:HIV:AIDS:capevert2015}.

Mathematical models, based on the underlying transmission 
mechanism of the HIV, can help the medical and scientific 
community to understand and anticipate its spread 
in different populations and evaluate the potential effectiveness 
of different approaches for bringing the epidemic 
under control \cite{Hyman:Stanley:1988}. 
Several mathematical models have been proposed 
for HIV/AIDS transmission dynamics: see, e.g., 
\cite{Anderson:JAIDS:1988,Anderson:etall:PTRSL:1989,Bhunu:etall:ActaBiothe:2009,%
Bhunu:etall:JMMA:2011,Cai:etall:JCAM:2009,Granich:etall:TheLancet:2009,%
Greenhalgh:etall:IMAJMAMB:2001,Hyman:Stanley:1988,Joshi:etall:MBioEng:2008,%
MayAnderson:Nature:1987,Musgrave:etall:MBioENg:2009} 
and references cited therein. 

In this paper, based on \cite{SilvaTorres:TBHIV:2015},
we propose and analyse a mathematical model for the 
transmission dynamics of HIV and AIDS. Our aim is to show 
that a simpler mathematical model than the one of
\cite{SilvaTorres:TBHIV:2015} can help to clarify some 
of the essential relations between epidemiological factors 
and the overall pattern of the AIDS epidemic, 
as it was stated in \cite{MayAnderson:Nature:1987}. 
Our model considers ART treatment of HIV-infected individuals 
with and without AIDS symptoms. The HIV-infected individuals 
with no AIDS symptoms that start ART treatment move into 
the class of individuals that respect carefully ART treatment 
and stay in what we call a \emph{chronic} stage of the infection. 
We assume that individuals in the chronic class have the same expectancy 
of life as non HIV infected individuals. Our model is then calibrated 
to the cumulative cases of infection by HIV and AIDS in Cape Verde, 
from 1987 to 2014, as reported in \cite{report:HIV:AIDS:capevert2015}. 
We show that our model predicts well such reality. 
A sensitivity analysis is then carried out, for some of the model 
parameters, and some conclusions for Cape Verde and worldwide inferred. 

The paper is organized as follows. The model is formulated 
and analysed in Section~\ref{sec:SICA:model}. 
In Section~\ref{sec:global:stab}, we prove the global stability 
of the unique endemic equilibrium whenever the basic reproduction 
number is greater than one. In Section~\ref{sec:sensitivity}, 
we do a sensitivity analysis of the basic reproduction number. 
Finally, in Section~\ref{sec:model:CapeVerde} we apply our model 
to HIV/AIDS Cape Verde infection data from the period of 1987 to 2014 
and some conclusions are derived in Section~\ref{sec:conclusion}.


\section{The SICA model for HIV/AIDS transmission}
\label{sec:SICA:model}

In this section we propose and analyse a mathematical model 
for HIV/AIDS transmission with varying population size 
in a homogeneously mixing population. 
The model is based on that of \cite{SilvaTorres:TBHIV:2015} 
and subdivides the human population into four mutually-exclusive 
compartments: susceptible individuals ($S$); 
HIV-infected individuals with no clinical symptoms of AIDS 
(the virus is living or developing in the individuals 
but without producing symptoms or only mild ones) 
but able to transmit HIV to other individuals ($I$); 
HIV-infected individuals under ART treatment (the so called 
chronic stage) with a viral load remaining low ($C$); 
and HIV-infected individuals with AIDS clinical symptoms ($A$).
The total population at time $t$, denoted by $N(t)$, is given by
\begin{equation*}
N(t) = S(t) + I(t) + C(t) + A(t).
\end{equation*} 
The susceptible population is increased by the recruitment 
of individuals into the population, assumed susceptible,  
at a rate $\Lambda$. All individuals suffer from natural death, 
at a constant rate $\mu$. Although individuals in the chronic stage, 
with a low viral load and under ART treatment, can still transmit 
HIV infection, as ART greatly reduces the risk of transmission 
and individuals that take ART treatment correctly are aware 
of their health status, we assume that individuals in the class $C$ 
do not have risky behaviours for HIV transmission and do not 
transmit HIV virus. We also assume that individuals 
with AIDS clinical symptoms $A$ are responsible 
and do not have any behaviour that can transmit HIV 
infection or, in other cases, are too sick to have a risky behaviour.  
Based on these two assumptions, susceptible individuals acquire HIV infection
by following effective contact with individuals in the class $I$ 
at a rate $\lambda = \beta \frac{I}{N}$, where $\beta$ is the effective 
contact rate for HIV transmission. We assume that HIV-infected individuals 
with and without AIDS symptoms have access to ART treatment. 
HIV-infected individuals with no AIDS symptoms $I$ progress to the class 
of individuals with HIV infection under ART treatment $C$ at a rate $\phi$, 
and HIV-infected individuals with AIDS symptoms are treated for HIV at rate $\alpha$.
We assume that an HIV-infected individual with AIDS symptoms $A$ 
that starts treatment moves to the class of HIV-infected individuals $I$ 
and he will move to the chronic class $C$ only if the treatment is maintained. 
HIV-infected individuals with no AIDS symptoms $I$ that do not take 
ART treatment progress to the AIDS class $A$ at rate $\rho$. We assume 
that only HIV-infected individuals with AIDS symptoms $A$ 
suffer from an AIDS induced death, at a rate $d$. These assumptions 
are translated in the following mathematical model:
\begin{equation}
\label{eq:model}
\begin{cases}
\dot{S}(t) = \Lambda - \beta \frac{I(t)}{N(t)} S(t) - \mu S(t),\\[0.2 cm]
\dot{I}(t) = \beta \frac{I(t)}{N(t)} S(t) - (\rho + \phi + \mu)I(t)
+ \alpha A(t)  + \omega C(t), \\[0.2 cm]
\dot{C}(t) = \phi I(t) - (\omega + \mu)C(t),\\[0.2 cm]
\dot{A}(t) =  \rho \, I(t) - (\alpha + \mu + d) A(t).
\end{cases}
\end{equation}
The epidemiological scheme of \eqref{eq:model} is given in Figure~\ref{epid:scheme}.
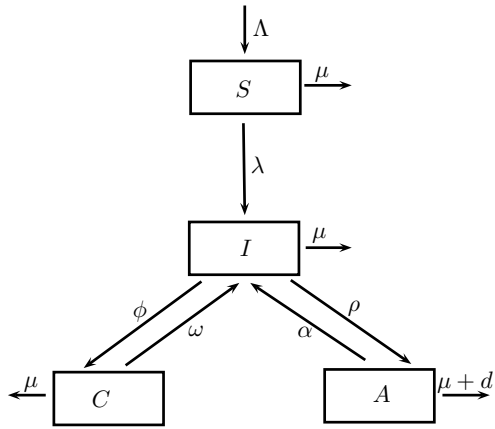
\begin{figure}[!htb]
\centering
\scalebox{0.9}{
\begin{pspicture}(0,-3.11)(7.5228124,3.13)
\psframe[linewidth=0.04,dimen=outer](4.4209375,2.31)(2.7609375,1.49)
\psframe[linewidth=0.04,dimen=outer](4.4009376,-0.07)(2.7409375,-0.89)
\psframe[linewidth=0.04,dimen=outer](2.4009376,-2.29)(0.7409375,-3.11)
\psframe[linewidth=0.04,dimen=outer](6.4009376,-2.25)(4.7409377,-3.07)
\psline[linewidth=0.04cm,arrowsize=0.05291667cm 2.0,arrowlength=1.4,arrowinset=0.4]{->}(3.5609374,1.37)(3.5809374,0.01)
\psline[linewidth=0.04cm,arrowsize=0.05291667cm 2.0,arrowlength=1.4,arrowinset=0.4]{->}(2.9409375,-0.97)(1.2009375,-2.25)
\psline[linewidth=0.04cm,arrowsize=0.05291667cm 2.0,arrowlength=1.4,arrowinset=0.4]{->}(4.2609377,-0.95)(6.0409374,-2.19)
\psline[linewidth=0.04cm,arrowsize=0.05291667cm 2.0,arrowlength=1.4,arrowinset=0.4]{->}(1.8209375,-2.21)(3.5009375,-0.99)
\psline[linewidth=0.04cm,arrowsize=0.05291667cm 2.0,arrowlength=1.4,arrowinset=0.4]{->}(5.3609376,-2.13)(3.6609375,-0.99)
\psline[linewidth=0.04cm,arrowsize=0.05291667cm 2.0,arrowlength=1.4,arrowinset=0.4]{->}(3.5809374,3.11)(3.5809374,2.39)
\psline[linewidth=0.04cm,arrowsize=0.05291667cm 2.0,arrowlength=1.4,arrowinset=0.4]{->}(4.4609375,1.93)(5.1609373,1.93)
\psline[linewidth=0.04cm,arrowsize=0.05291667cm 2.0,arrowlength=1.4,arrowinset=0.4]{->}(6.5009375,-2.65)(7.2009373,-2.65)
\psline[linewidth=0.04cm,arrowsize=0.05291667cm 2.0,arrowlength=1.4,arrowinset=0.4]{->}(0.6409375,-2.65)(0.0809375,-2.65)
\psline[linewidth=0.04cm,arrowsize=0.05291667cm 2.0,arrowlength=1.4,arrowinset=0.4]{->}(4.4809375,-0.47)(5.1609373,-0.47)
\usefont{T1}{ptm}{m}{n}
\rput(3.5623438,1.88){$S$}
\usefont{T1}{ptm}{m}{n}
\rput(3.6023438,-0.48){$I$}
\usefont{T1}{ptm}{m}{n}
\rput(1.4623437,-2.7){$C$}
\usefont{T1}{ptm}{m}{n}
\rput(5.6023436,-2.64){$A$}
\usefont{T1}{ptm}{m}{n}
\rput(3.8223438,2.78){$\Lambda$}
\usefont{T1}{ptm}{m}{n}
\rput(4.722344,2.1){$\mu$}
\usefont{T1}{ptm}{m}{n}
\rput(4.682344,-0.28){$\mu$}
\usefont{T1}{ptm}{m}{n}
\rput(0.42234376,-2.48){$\mu$}
\usefont{T1}{ptm}{m}{n}
\rput(6.8323436,-2.46){$\mu+d$}
\usefont{T1}{ptm}{m}{n}
\rput(3.7723436,0.76){$\lambda$}
\usefont{T1}{ptm}{m}{n}
\rput(2.0323439,-1.4){$\phi$}
\usefont{T1}{ptm}{m}{n}
\rput(2.8623438,-1.7){$\omega$}
\usefont{T1}{ptm}{m}{n}
\rput(4.472344,-1.72){$\alpha$}
\usefont{T1}{ptm}{m}{n}
\rput(5.1923437,-1.36){$\rho$}
\end{pspicture}}
\caption{Epidemiological scheme of the SICA mathematical model \eqref{eq:model}.}
\label{epid:scheme}
\end{figure}

Consider the biologically feasible region
\begin{equation*}
\Omega = \left\{ (S, I, C, A) \in {\mathbb{R}}^4_{+0} \, | \,  0
\leq S + I + C + A \leq \frac{\Lambda}{\mu} \right\}.
\end{equation*}
Using a standard comparison theorem (see \cite{Lakshmikantham:1989}) 
one can easily show that $N(t) \leq \frac{\Lambda}{\mu}$ 
if $N(0) \leq \frac{\Lambda}{\mu}$. Thus, the region $\Omega$ 
is positively invariant. Hence, it is sufficient to consider
the dynamics of the flow generated by \eqref{eq:model} in $\Omega$.
In this region, the model is epidemiologically and mathematically 
well posed in the sense of \cite{Hethcote:2000}. In other words,
every solution of the model \eqref{eq:model} with initial conditions
in $\Omega$ remains in $\Omega$ for all $t > 0$. For the persistence 
of the host population and disease we refer the reader to 
\cite[Section 3.1.1]{SilvaTorres:TBHIV:2015}. 
Precisely, the disease is persistent in the population 
if the fraction of the infected and AIDS cases is bounded away from zero. 
If the population dies out and the fraction of the infected and AIDS 
remains bounded away from zero, the disease still remains persistent 
in the population. The model \eqref{eq:model} 
has a disease-free equilibrium given by
\begin{equation*}
\Sigma_0 = \left(S^0, I^0, C^0, A^0 \right)
= \left(\frac{\Lambda}{\mu},0, 0,0  \right).
\end{equation*}
The basic reproduction number $R_0$, which represents 
the expected average number of new HIV infections
produced by a single HIV-infected individual when in contact 
with a completely susceptible population, is easily computed 
by the van den Driessche and Watmough approach \cite{van:den:Driessche:2002}:
\begin{equation}
\label{eq:R0}
R_0 = \frac{\beta \xi_1 \xi_2 }{\mu \left[ \xi_2 (\rho + \xi_1)
+ \xi_1 \phi + \rho d \right] + \rho \omega d} 
= \frac{\mathcal{D}}{\mathcal{N}},
\end{equation}
where 
$$
\xi_1 = \alpha + \mu + d \ \text{ and } \ 
\xi_2 = \omega + \mu.
$$ 

\begin{lemma}[See \cite{SilvaTorres:TBHIV:2015}]
The disease free equilibrium $\Sigma_0$ 
is globally asymptotically stable if $R_0 < 1$.
If $R_0 > 1$, then the disease free equilibrium $\Sigma_0$ is unstable.
\end{lemma}

To find conditions for the existence of an equilibrium 
$\Sigma^* = \left(S^*, I^*, C^*, A^* \right)$ for which HIV is endemic
in the population (i.e., at least one of $I^*$, $C^*$ or $A^*$ is nonzero),
the equations in \eqref{eq:model} are solved in terms of the force of infection
at steady-state ($\lambda^*$), given by
\begin{equation}
\label{eq:lambdaH:ast}
\lambda^* = \frac{\beta I^*}{N^*}.
\end{equation}
Setting the right-hand sides of the model \eqref{eq:model}
to zero (and noting that $\lambda = \lambda^*$ at equilibrium) gives
\begin{equation}
\label{eq:end:equil:HIV}
S^* = \frac{\Lambda}{\lambda^* + \mu}, \quad 
I^*=\frac{\lambda^* \Lambda \xi_1 \xi_2}{D}, \quad C^*
=\frac{\phi \lambda^* \Lambda \xi_1}{D}, \quad A^*
=\frac{\rho \lambda^* \Lambda \xi_2}{D}
\end{equation}
with
$D = (\lambda^* + \mu)\left[\mu \left(\xi_2 (\rho_1 + \xi_1)
+ \xi_1 \phi + \rho d \right) + \rho \omega d\right]$.
Using \eqref{eq:end:equil:HIV} in the expression
for $\lambda^*$ in \eqref{eq:lambdaH:ast} shows that
the nonzero (endemic) equilibria of the model satisfies
\begin{equation*}
\lambda^* = \mu (R_0 -1).
\end{equation*}
The force of infection at the steady-state $\lambda^*$ is positive
only if $R_0 > 1$. Therefore, system \eqref{eq:model} has a unique 
endemic equilibrium whenever $R_0 > 1$. This is in agreement 
with \cite[Lemma 3.7]{SilvaTorres:TBHIV:2015}. 


\section{Global stability of the endemic equilibrium}
\label{sec:global:stab}

In this section, we prove the global stability of the endemic equilibrium $\Sigma^*$. 
We assume that each variable $S$, $I$, $C$ and $A$ 
represents the corresponding proportion of the total population 
and consider the following system:
\begin{equation}
\label{eq:model:2}
\begin{cases}
\dot{S}(t) = \Lambda - \beta I(t) S(t) - \mu S(t),\\[0.2 cm]
\dot{I}(t) = \beta I(t) S(t) - (\rho + \phi + \mu)I(t)
+ \alpha A(t)  + \omega C(t), \\[0.2 cm]
\dot{C}(t) = \phi I(t) - (\omega + \mu)C(t),\\[0.2 cm]
\dot{A}(t) =  \rho \, I(t) - (\alpha + \mu + d) A(t).
\end{cases}
\end{equation}
Let us start by defining the region 
\begin{equation*}
\Omega_0 = \left\{ (S, I, C, A) \in \Omega \, | \,  I = C = A = 0 \right\}.
\end{equation*}
Consider the following Lyapunov function:
\begin{multline*}
V = S - S^* - S^* \ln\left(\frac{S}{S^*} \right) + I - I^* - I^* \ln\left(\frac{I}{I^*} \right) 
+ \frac{\omega}{\xi_2} \left(C - C^* - C^* \ln\left(\frac{C}{C^*} \right)\right)\\
+ \frac{\alpha}{\xi_1} \left(A - A^* - A^* \ln\left(\frac{A}{A^*} \right)\right).
\end{multline*}
Differentiating $V$ with respect to time gives
\begin{equation*}
\dot{V} = \left(1-\frac{S^*}{S}\right)\dot{S} + \left(1-\frac{I^*}{I}\right)\dot{I} 
+ \frac{\omega}{\xi_2} \left(1-\frac{C^*}{C}\right)\dot{C} + \frac{\alpha}{\xi_1} 
\left(1-\frac{A^*}{A}\right)\dot{A}.
\end{equation*}
Let $\xi_3 = \rho + \phi + \mu$. Substituting the expressions 
for the derivatives in $\dot{V}$, it follows from 
\eqref{eq:model:2} that
\begin{multline}
\label{eq:difV:1}
\dot{V} 
= \left(1-\frac{S^*}{S}\right)\left[ \Lambda - \beta I S -\mu S \right] 
+ \left(1-\frac{I^*}{I}\right)\left[ \beta I S - \xi_3 I + \alpha A + \omega C \right]\\
+ \frac{\omega}{\xi_2}  \left(1-\frac{C^*}{C}\right)\left[  \phi I -  \xi_2 C \right] 
+ \frac{\alpha}{\xi_1} \left(1-\frac{A^*}{A}\right)\left[  \rho I - \xi_1 A \right].
\end{multline}
Using the relation $\Lambda = \beta I^* S^* + \mu S^*$, 
we have from the first equation of system \eqref{eq:model:2} 
at steady-state that \eqref{eq:difV:1} can be written as
\begin{multline*}
\dot{V} = \left(1-\frac{S^*}{S}\right)\left[  \beta I^* S^* + \mu S^* - \beta I S -\mu S \right]
+ \left(1-\frac{I^*}{I}\right)\left[ \beta I S - \xi_3 I + \alpha A + \omega C \right]\\
+ \frac{\omega}{\xi_2} \left(1-\frac{C^*}{C}\right)\left[  \phi I -  \xi_2 C \right]  
+ \frac{\alpha}{\xi_1} \left(1-\frac{A^*}{A}\right)\left[  \rho I - \xi_1 A \right],
\end{multline*}
which can then be simplified to
\begin{multline*}
\dot{V} = \beta I^* S^* \left(1-\frac{S^*}{S} \right)  
+ \mu S^* \left(2- \frac{S}{S^*} 
- \frac{S^*}{S} \right) + S^* \beta I
- \xi_3 I +  \frac{\omega}{\xi_2} \phi I 
+ \frac{\alpha}{\xi_1} \rho I \\
-\frac{I^*}{I} \left[ \beta I S 
- \xi_3 I + \alpha A + \omega C \right]
-\frac{\omega}{\xi_2}\frac{C^*}{C} \left( \phi I -  \xi_2 C \right) 
-\frac{\alpha}{\xi_1} \frac{A^*}{A} \left( \rho I - \xi_1 A \right).
\end{multline*}
Using the relations at the steady state
\begin{equation*}
I^* = \frac{\beta I^* S^*  + \alpha A^*  + \omega C^*}{\xi_3}, 
\quad C^* = \frac{\phi I^*}{\xi_2}, 
\quad A^* = \frac{\rho I^*}{\xi_1},
\end{equation*} 
we have
\begin{multline*}
\dot{V} = \beta I^* S^* \left(1-\frac{S^*}{S} \right)  
+ \mu S^* \left(2- \frac{S}{S^*} - \frac{S^*}{S} \right)
-\frac{I^*}{I} \left( \beta I S - \xi_3 I + \alpha A + \omega C \right)\\
-\frac{\omega}{\xi_2} \frac{C^*}{C} \left( \phi I -  \xi_2 C \right) 
- \frac{\alpha}{\xi_1} \frac{A^*}{A} \left( \rho I - \xi_1 A \right).
\end{multline*}
Simplifying further the previous equations, we finally have
\begin{equation*}
\begin{split}
\dot{V} &= \left( \beta I^* S^* + \mu S^* \right) \left(2-\frac{S^*}{S} 
- \frac{S}{S^*}  \right)\\
&\quad + \omega C^* \left(1 - \frac{C}{C^*} \frac{I^*}{I} \right) 
+ \alpha A^* \left( 1 - \frac{A}{A^*} \frac{I^*}{I}\right)  \\
& \quad + \frac{\omega \phi}{\xi_2} I^* \left(1 -\frac{C^*}{C} \frac{I}{I^*}  \right) 
+ \frac{\alpha \rho}{\xi_1} I^* \left(1 - \frac{A^*}{A}  \frac{I}{I^*} \right),
\end{split}
\end{equation*}
which is less than or equal to zero by the arithmetic 
mean-geometric mean inequality. Therefore, we have 
$\dot{V} \leq 0$ with equality holding only 
if $S=S^*$ and $\frac{I}{I^*} = \frac{C}{C^*}=\frac{A}{A^*}$. 
By LaSalle's invariance principle \cite{LaSalle}, 
the omega limit set of each solution lies in an invariant set
contained in 
$$
\mathcal{P} = \left\{ (S, I, C, A) \, | \, S=S^*, 
\frac{I}{I^*} = \frac{C}{C^*}=\frac{A}{A^*} \right\}.
$$ 
Since $S$ must remain constant at $S^*$, $\dot{S}$ is zero. 
This implies that $I=I^*$, making $\frac{I}{I^*}$ equal to one. 
Thus, $C=C^*$ and $A=A^*$. Hence, the only invariant set contained 
in $\mathcal{P}$ is the singleton $\{ \Sigma^*\}$. This shows 
that the endemic equilibrium $\Sigma^*$ is the only 
solution that intersects $\Omega$. We have just proved 
that if $R_0 > 1$, then the unique endemic equilibrium 
$\Sigma^*$ is globally asymptotically stable 
on $\Omega \backslash \Omega_0$. This result
is summarized  in the following theorem. 

\begin{theorem}
\label{theorem:global:stab:EE}
The unique endemic equilibrium $\Sigma^*$ of model \eqref{eq:model:2} 
is globally asymptotically stable in $\Omega \backslash \Omega_0$ 
whenever $R_0 > 1$. 
\end{theorem}


\section{Sensitivity of the basic reproduction number}
\label{sec:sensitivity}

The sensitivity of the basic reproduction number $R_0$
is an important issue because it determines the model robustness 
to parameter values. Let us examine the sensitivity of $R_0$ 
with respect to some parameters. We start by examining 
the sensitivity of $R_0$ with respect to $\beta$.

\begin{proposition}
\label{prop:role:beta:red:r0}
The basic reproduction number $R_0$ increases with $\beta$.
\end{proposition}

\begin{proof}
The results follows immediately from the fact that
$$
\frac{\partial R_0}{\partial \beta} > 0
$$
for any value of the parameters.
\end{proof}

The partial derivatives of $R_0$ with respect 
to the treatment rate for HIV-infected individuals 
with no symptoms of AIDS and HIV-infected individuals 
with AIDS symptoms, $\phi$ and $\alpha$, respectively, 
are given by 
\begin{equation}
\label{eq:difR0:phi:alpha}
\frac{\partial R_0}{\partial \phi} 
= -\frac{\beta \mu \xi_1^2 \xi_2}{\mathcal{D}^2}, 
\qquad \frac{\partial R_0}{\partial \alpha} 
= \frac{\beta \rho \xi_2^2 (\mu + d) }{\mathcal{D}^2}.
\end{equation}
It follows from \eqref{eq:difR0:phi:alpha} that 
$\frac{\partial R_0}{\partial \phi} < 0$ and 
$\frac{\partial R_0}{\partial \alpha} > 0$ 
for any value of the parameters. Therefore, 
the treatment of HIV-infected individuals with 
no symptoms of AIDS, $I$, has always a positive 
impact in reducing the HIV burden. On the other hand, 
the treatment of HIV-infected individuals increases 
the number of HIV-infected individuals, explained
because when they start ART treatment their 
expectancy of life increases and there will be more 
people living with HIV-infection with no AIDS symptoms. 
This is summarized in the following result.

\begin{proposition}
The basic reproduction number $R_0$ decreases 
with $\phi$ and increases with $\alpha$. 
\end{proposition}

The partial derivatives of $R_0$ with respect 
to the default treatment rates $\rho$ and $\omega$ 
are given, respectively, by
\begin{equation*}
\frac{\partial R_0}{\partial \rho} 
= -\frac{\beta \xi_1 \xi_2^2 (\mu + d) }{\mathcal{D}^2}, 
\qquad \frac{\partial R_0}{\partial \omega} 
= \frac{\beta \phi \mu \xi_1^2 }{\mathcal{D}^2}.
\end{equation*}
Thus, the basic reproduction number increases with the default 
treatment rate of individuals in the chronic class $\omega$ 
and decreases with the default treatment rate of individuals 
in the class $I$. This makes sense because the lack of ART 
treatment implies a progression to AIDS disease. 
This is stated in the following result.

\begin{proposition}
The basic reproduction number $R_0$ decreases 
with $\rho$ and increases with $\omega$. 
\end{proposition}

The sensitivity of a variable (in our case of interest, $R_0$) 
with respect to model parameters is sometimes measured 
by the so called \emph{sensitivity index}.

\begin{definition}[cf. \cite{Chitnis,Kong}]
\label{def:sense}
The normalized forward sensitivity index of a variable 
$\upsilon$ that depends differentiably 
on a parameter $p$ is defined by
\begin{equation}
\label{eq:def:sense}
\Upsilon_{p}^{\upsilon} 
:= \frac{\partial \upsilon}{\partial p} \times \frac{p}{|\upsilon|}.
\end{equation}
\end{definition}

\begin{remark}
To the most sensitive parameter $p$ it corresponds 
a normalized forward sensitivity index of one or minus one,
that is, $\Upsilon_{p}^{\upsilon} = \pm 1$. 
If $\Upsilon_{p}^{\upsilon} = + 1$,
then an increase (decrease) of $p$ by $x\%$ increases (decreases) 
$\upsilon$ by $x\%$; if $\Upsilon_{p}^{\upsilon} = - 1$, 
then an increase (decrease) of $p$ by $x \%$ decreases (increases)
$\upsilon$ by $x\%$. 
\end{remark}

From \eqref{eq:R0} and Definition~\ref{def:sense}, 
it is easy to derive the normalized forward 
sensitivity index of $R_0$ with respect to $\beta$.

\begin{proposition}
The normalized forward sensitivity index 
of $R_0$ with respect to $\beta$ is $1$:
$\Upsilon_{\beta}^{R_0} = 1$.
\end{proposition}

\begin{proof}
It is a direct consequence of \eqref{eq:R0} and \eqref{eq:def:sense}. 
\end{proof}

The sensitivity index of $R_0$ with respect to $\phi$, 
$\rho$, $\alpha$ and $\omega$ is given, respectively, by
$$
\Upsilon_{\phi}^{R_0} = -\phi\frac{\mu C_3}{\mathcal{D}}, 
\quad \Upsilon_{\rho}^{R_0} = -\rho\frac{C_2 (\mu + d)}{\mathcal{D}}, 
\quad \Upsilon_{\alpha}^{R_0} = \alpha \frac{\rho C_2 (\mu + d)}{C_3 \mathcal{D}}, 
\quad \Upsilon_{\omega}^{R_0} = \omega \frac{\phi \mu C_3}{C_2 \mathcal{D}}. 
$$
In Section~\ref{sec:model:CapeVerde} we compute the previous 
sensitivity indexes for the Cape Verde HIV/AIDS data. 


\section{Model application to Cape Verde HIV/AIDS data}
\label{sec:model:CapeVerde}

In this section, we calibrate our model \eqref{eq:model} 
to the cumulative cases of infection by HIV and AIDS 
from 1987 to 2014 and we show that it predicts well this reality. 
In Table~\ref{table:realdataCapeVerde}, the cumulative cases 
of infection by HIV and AIDS in Cape Verde are depicted 
for the years 1987--2014 \cite{report:HIV:AIDS:capevert2015}. 
\begin{table}[!htb]
\centering
\begin{tabular}{l  l  l  l l l l l l l l} \hline \hline
{\small{Year}} &  {\small{1987}}   &  {\small{1988}} & {\small{1989}} & {\small{1990}} & {\small{1991}} 
& {\small{1992}} & {\small{1993}} & {\small{1994}} & {\small{1995}} & {\small{1996}}\\ \hline
{\small{HIV/AIDS}} & {\small{61}} & {\small{107}}  &  {\small{160}} &  {\small{211}} &  {\small{244}} 
& {\small{303}} &  {\small{337}} &  {\small{358}} & {\small{395}}  & {\small{432}}\\ \hline \hline
{\small{Year}} & {\small{1997}} &  {\small{1998}}  &  {\small{1999}} & {\small{2000}} & {\small{2001}} 
& {\small{2002}} & {\small{2003}} & {\small{2004}} & {\small{2005}}  & {\small{2006}}\\
{\small{HIV/AIDS}} & {\small{471}} & {\small{560}} & {\small{660}}  &  {\small{779}} &  {\small{913}} 
& {\small{1064}} &  {\small{1233}} &  {\small{1493}}  &  {\small{1716}}  & {\small{2015}}\\ \hline \hline
{\small{Year}} & {\small{2007}} & {\small{2008}} &  {\small{2009}} & {\small{2010}} & {\small{2011}} 
& {\small{2012}} & {\small{2013}} & {\small{2014}} & &\\
{\small{HIV/AIDS}}  & {\small{2334}} & {\small{2610}} & {\small{2929}} & {\small{3340}}  
&  {\small{3739}} &  {\small{4090}} & {\small{4537}} & {\small{4946}} & &\\ \hline \hline
\end{tabular}
\caption{Cumulative cases of infection by HIV/AIDS in Cape Verde 
in the period 1987--2014 \cite{report:HIV:AIDS:capevert2015}.}
\label{table:realdataCapeVerde}
\end{table}
Taking into account the data from Table~\ref{table:realdataCapeVerde}, 
we estimate empirically the value of the HIV transmission rate to be 
$\beta = 0.857$. Searching for the optimal value of $\beta$ 
in the interval $[0.8, 0.9]$ with respect to the $l_2$ norm,
we find that the optimal value of $\beta$ is $\beta =  0.866$, 
with an error equal to 0.507\% individuals per year with respect 
to the initial total population. In all our simulations we choose
$\beta =  0.866$. Moreover, we consider the initial conditions 
\eqref{eq:initcond:CV} based on 
\cite{report:HIV:AIDS:capevert2015,url:worlbank:capevert}: 
\begin{equation}
\label{eq:initcond:CV}
S_0 = S(0) = 338923 \, , 
\quad I_0 = I(0) = 61\, , 
\quad C_0 = C(0) = 0\, , 
\quad A_0 = A(0) = 0.
\end{equation}
The values of the parameters $\rho = 0.1$ and $\alpha = 0.33$ are taken from 
\cite{Sharomi:MathBio:2008} and \cite{Bhunu:BMB:2009:HIV:TB}, respectively. 
We assume that after one year, the HIV infected individuals $I$ that are under
ART treatment have a low viral load \cite{Perelson} and, therefore, 
are transferred to the class $C$. In agreement, we take $\phi=1$. 
It is well known that taking ART therapy is a long-term commitment.
In our simulations we assume that the default treatment rate for $C$ 
individuals is approximately 11 years ($1/\omega$ years to be precise).
Following the World Bank data \cite{url:worlbank:capevert},
the recruitment and the natural death rates are assumed to take the values 
$\Lambda = 10724$ and $\mu = 1/69.54$. Finally, the AIDS induced death 
rate is assume to be $d = 1$ based on \cite{ZwahlenEggerUNAIDS}. 
All the considered parameter values are resumed 
in Table~\ref{table:parameters:HIV:CV}. 
\begin{table}[!htb]
\centering
\begin{tabular}{l  p{6.5cm} l l}
\hline \hline
{\small{Symbol}} &  {\small{Description}} & {\small{Value}} & {\small{References}}\\
\hline
{\small{$N(0)$}} & {\small{Initial population}} & {\small{$338 984$}}  
& {\small{\cite{url:worlbank:capevert}}}\\
{\small{$\Lambda$}} & {\small{Recruitment rate}} & {\small{$10724$}}  
& {\small{\cite{url:worlbank:capevert} }}\\
{\small{$\mu$}} & {\small{Natural death rate}} & {\small{$1/69.54$}} 
& {\small{\cite{url:worlbank:capevert} }}\\
{\small{$\beta$}} & {\small{HIV transmission rate}} & {\small{$0.866$}} & {\small{Estimated}}\\
{\small{$\phi$}} & {\small{HIV treatment rate for $I$ individuals}} &  {\small{$1$}} 
& {\small{\cite{Perelson}}} \\
{\small{$\rho$}} & {\small{Default treatment rate for $I$ individuals}}
& {\small{$0.1 $}} & {\small{\cite{Sharomi:MathBio:2008}}}\\
{\small{$\alpha$}} & {\small{AIDS treatment rate}}
& {\small{$0.33 $}} & {\small{\cite{Bhunu:BMB:2009:HIV:TB}}}\\
{\small{$\omega$}} & {\small{Default treatment rate for $C$ individuals}}
& {\small{$0.09$}} & {\small{Assumed}}\\
{\small{$d$}} & {\small{AIDS induced death rate}} & {\small{$1$}} 
& {\small{\cite{ZwahlenEggerUNAIDS}}}\\
\hline \hline
\end{tabular}
\caption{Parameters of the HIV/AIDS model \eqref{eq:model} for Cape Verde.}
\label{table:parameters:HIV:CV}
\end{table}

In Figure~\ref{fig:model:fit}, we observe that model \eqref{eq:model} 
fits the real data reported in Table~\ref{table:realdataCapeVerde}. 
The cumulative cases described by model \eqref{eq:model} are given by 
$I(t) + C(t) + A(t) + \mu\left( I(t) + C(t)\right) + (\mu + d) \left( A(t)\right)$ 
for $t \in [0, 27]$, which corresponds to the interval of time between 
the years of 1987 ($t=0$) and 2014 ($t=27$). 
\begin{figure}[!htb]
\centering
\includegraphics[width=0.7\textwidth]{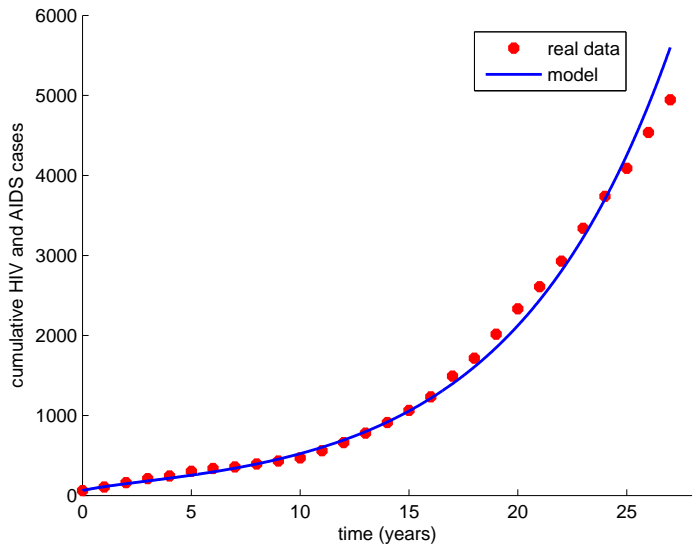}
\caption{Model \eqref{eq:model} fitting the data 
of cumulative cases of HIV and AIDS infection  
in Cape Verde between 1987 and 2014 \cite{report:HIV:AIDS:capevert2015}.}
\label{fig:model:fit}
\end{figure}
The parameter values from Table~\ref{table:parameters:HIV:CV} correspond 
to a basic reproduction number $R_0 = 3.8049$. 
Figure~\ref{fig:stab:EE} illustrates, numerically,
the global stability of the endemic equilibrium. For this reason,
we consider different initial conditions, in different regions of the plane,
sufficiently far a way from the endemic equilibrium
$(S^*, I^*, C^*, A^*) = (146000, 37894, 363040, 2818.7)$.  
\begin{figure}[!htb]
\centering
\subfloat[\footnotesize{$(S, I)$}]{\label{EEstabSI}
\includegraphics[width=0.45\textwidth]{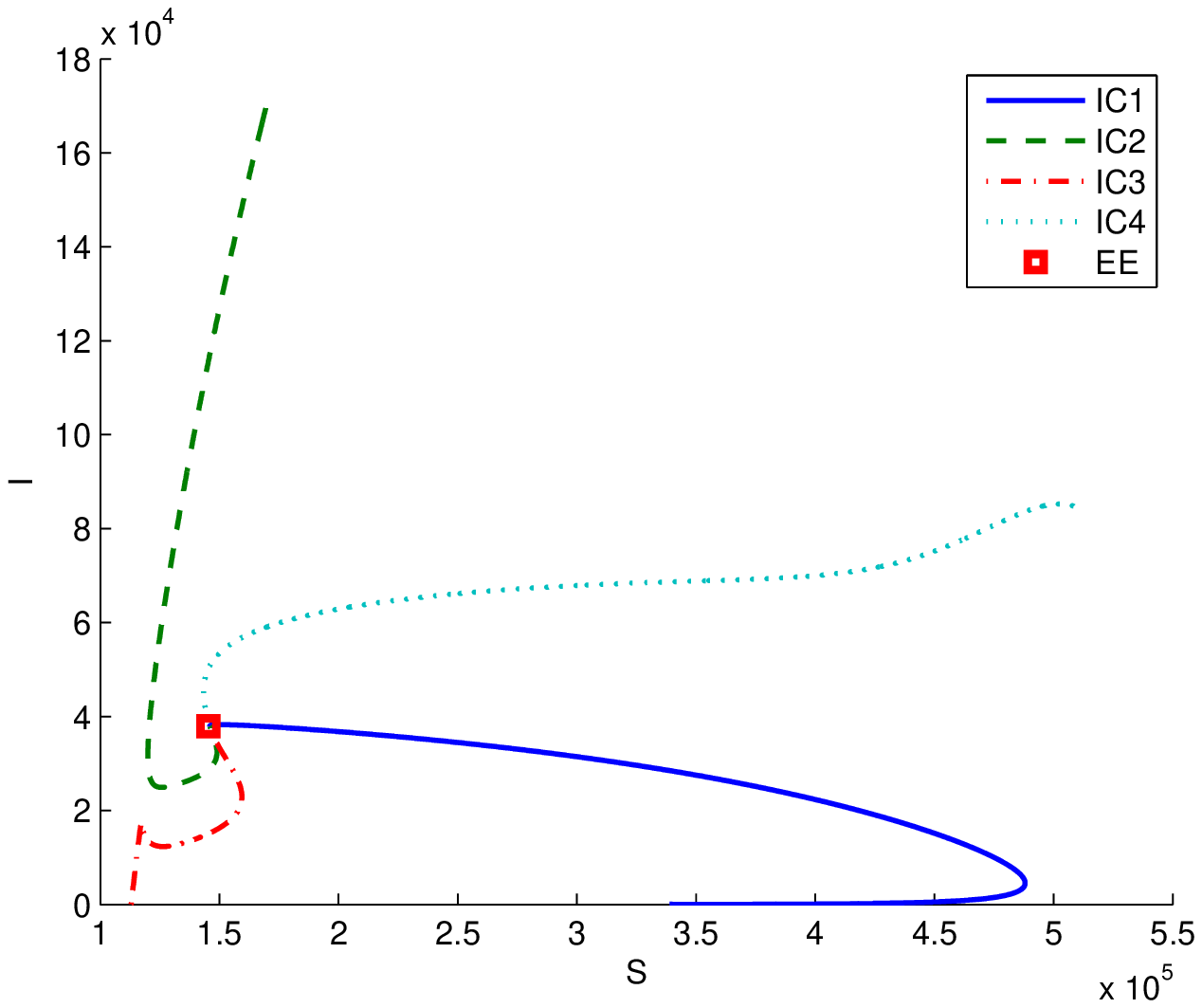}}
\subfloat[\footnotesize{$(C, A)$}]{\label{EEstabCA}
\includegraphics[width=0.45\textwidth]{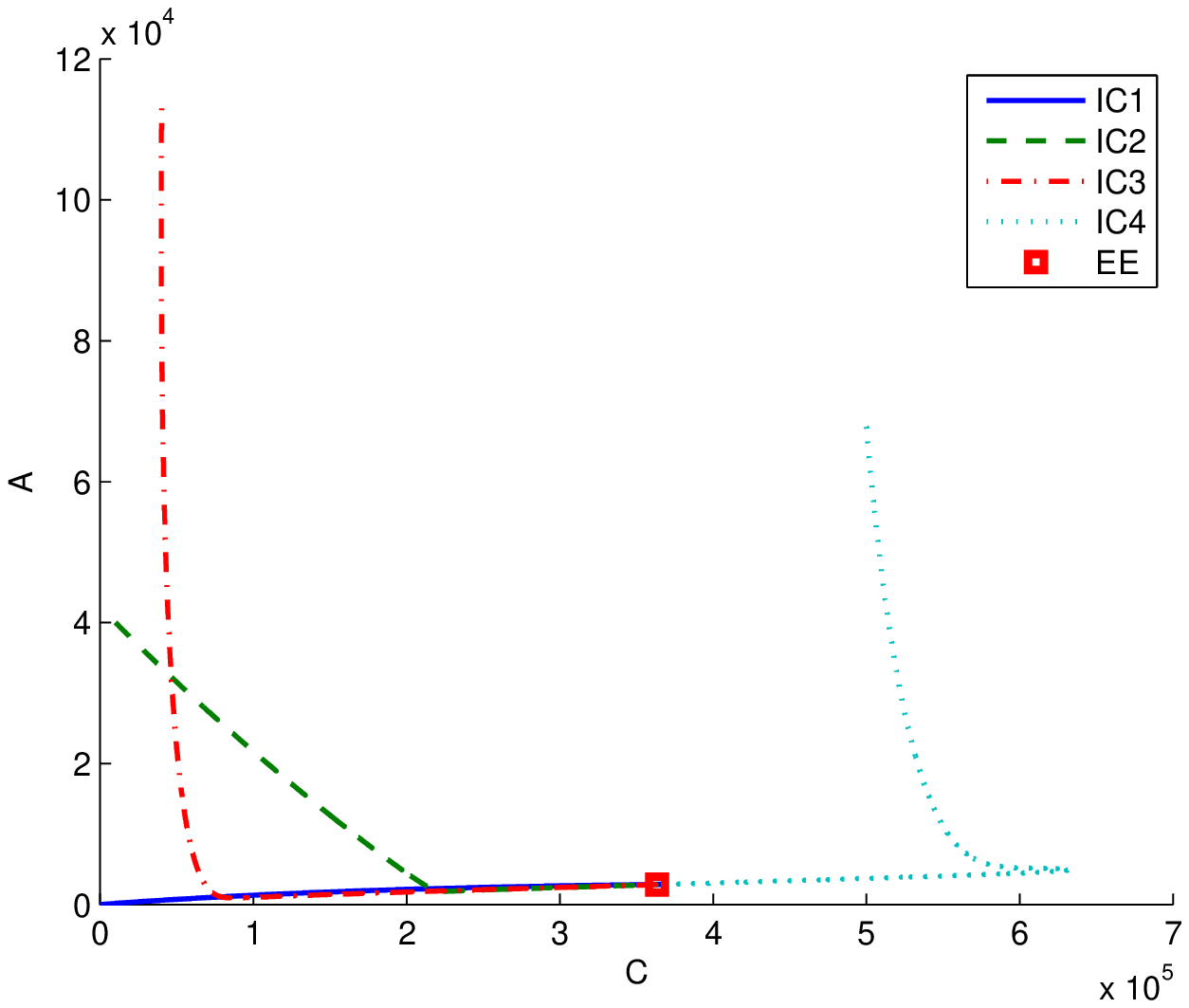}}
\caption{(a) Susceptible individuals \emph{versus} HIV-infected 
individuals with no clinical symptoms of AIDS. 
(b) HIV-infected individuals under treatment for HIV infection 
\emph{versus} HIV-infected individuals with AIDS clinical symptoms. 
Both figures consider initial conditions $IC1 = (S_0, I_0, C_0, A_0)$, 
$IC2 = (S_0/2, I_0+S_0/2, C_0+10^4, A_0+4 \times 10^4)$, 
$IC3 = (S_0/3, I_0, C_0+4 \times 10^4, A_0+S_0/3)$, 
and $IC4 = (3 S_0/2, I_0+S_0/4, C_0+5 \times 10^5, A_0+S_0/5)$.
EE denotes the endemic equilibrium 
$(S^*, I^*, C^*, A^*) = (146000, 37894, 363040, 2818.7)$.}
\label{fig:stab:EE}
\end{figure}
In Table~\ref{table:sensitivity:index}, we present the sensitivity index 
of parameters $\beta$, $\phi$, $\rho$, $\alpha$ and $\omega$ computed 
for the parameter values given in Table~\ref{table:parameters:HIV:CV}. 
\begin{table}[!htb]
\centering
\begin{tabular}{l l }
\hline \hline
{\small{Parameter}} & {\small{Sensitivity index}}\\
\hline
{\small{$\beta$}} & {\small{+1}} \\
{\small{$\phi$}} & {\small{$-0.6053$}}  \\
{\small{$\rho$}} & {\small{$-0.3315$}}\\
{\small{$\alpha$}} &  {\small{$+0.0814$}} \\
{\small{$\omega$}} &  {\small{$+0.5219$}} \\
\hline \hline
\end{tabular}
\caption{Sensitivity index of $R_0$ for parameter values 
given in Table~\ref{table:parameters:HIV:CV}.}
\label{table:sensitivity:index}
\end{table}


\section{Conclusion}
\label{sec:conclusion}

In this paper, we proposed a nonlinear mathematical model for HIV/AIDS 
transmission with varying population size in a homogeneously mixing population. 
The model was analysed for the particular situation where only HIV-infected 
individuals with no AIDS symptoms, which are not under ART treatment, 
transmit HIV virus. However, the global stability result can be easily 
extended to the case where individuals under ART treatment 
and individuals with AIDS disease also transmit HIV virus. 
We have shown that the proposed model describes very well 
the reality given by the data of HIV/AIDS infection 
in Cape Verde from 1987 to 2014. We conclude that, 
with the parameter values of Table~\ref{table:parameters:HIV:CV}, 
the endemic equilibrium for Cape Verde 
is $(S^*, I^*, C^*, A^*)=(146000, 37894, 363040, 2818.7)$, 
which is not an encouraging prediction. Indeed, the number of individuals 
with HIV-infection and AIDS disease is relatively big and increasing 
and does not converge to the UNAIDS worldwide goal of ending the AIDS epidemic by 2030 
\cite{AIDSnumbers2015}. From the sensitivity index of the basic reproduction number 
with respect to treatment rates for HIV-infected individuals with no AIDS symptoms, 
we may conclude that it is important to invest in providing conditions 
to HIV-infected individuals to maintain, correctly, the ART treatment. Analogously, 
HIV-infected individuals with AIDS disease should be supported in order 
to start ART treatment as soon as possible. All the measures that help 
to reduce the default treatment should be supported. Finally, it is essential
to reduce HIV transmission. We observed that the transmission rate 
is the parameter that most affects the basic reproduction number. 
 

\section*{Acknowledgments}

This research was partially supported by the
Portuguese Foundation for Science and Technology (FCT)
within projects UID/MAT/04106/2013 (CIDMA) and 
PTDC/EEI-AUT/2933/2014 (TOCCATA), co-funded by Project 3599 
-- Promover a Produ\c{c}\~ao Cient\'{\i}fica e Desenvolvimento
Tecnol\'ogico e a Constitui\c{c}\~ao de Redes Tem\'aticas (3599-PPCDT)
and FEDER funds through COMPETE 2020, Programa Operacional
Competitividade e Internacionaliza\c{c}\~ao (POCI).
Silva is also grateful to the FCT post-doc 
fellowship SFRH/BPD/72061/2010.
The authors are grateful to two anonymous Reviewers
for several comments and suggestions.



\end{document}